\documentclass[amsart,12pt]{amsart}
\usepackage{graphicx}
\usepackage{graphics}
\usepackage{amsmath, latexsym}
\usepackage{epsfig}
\usepackage{amssymb,latexsym}
\usepackage{graphicx}
\newtheorem{theorem}{Theorem}

\begin{document}
\title{\textbf{Connectivity Algorithm}}
\author{Krishnendra Shekhawat\\
PhD, Department of Mathematics, University of Geneva\\
\small\it{Email:\; krishnendra.iitd@gmail.com}}
\maketitle
{
\begin{abstract}
In this work, for the given adjacency matrix of a graph, we present an algorithm which checks the connectivity of a graph and computes all of its connected components. Also, it is mathematically proved that the algorithm presents all the desired results.
\end{abstract}



\section{Introduction}

A {\it simple graph} $G$ is an unweighted,  undirected graph containing
no self loops or multiple edges. 
Two vertices of a graph {\it 
adjacent} if there exist an edge between them.
An {\it adjacency matrix} for a simple graph is a matrix with rows and columns
labeled by graph vertices, with a $1$ or $0$ in position $(\ v_i,\ v_j\ )$ 
according to whether $v_i$ and $v_j$ are adjacent or not. 
For standard terminology and notation in graph theory, used throughout the text, we refer to West[4] and Gross[3].

In the literature, there exist many algorithms for obtaining the results mentioned in the abstract (e.g. refer [1-2]). But we did not find an algorithm, which alone can compute all the results.
Therefore,  in this paper, we present an algorithm which takes adjacency matrix as input and computes all above mentioned results. 

\section{Connectivity of a graph and its connected components}

In an undirected graph $G$, two vertices $u$ and $v$  
are {\it connected} if $G$ contains a path from $u$ to $v$. $G$ is called {\it connected}
if every pair of distinct vertices in the graph can be connected through some path.
$G$ is {\it disconnected} if it is not connected. 

A {\it connected component} of an undirected graph
is a subgraph in which any two vertices are connected to each other by paths  and to
which no more vertices or edges(from the larger graph) can be added while preserving  
its connectivity. 

\subsection{Algorithm D}

Let $G$ be a simple graph with $n$ vertices and $A$ be the corresponding  
adjacency matrix of $G$. Also consider that $A'$ is a null matrix
of order $n$ and $G^{\prime}$ is its corresponding graph. At the end of the algorithm, we eill find that graphs $G$ and $G^{\prime}$ are same, only difference is in the labelling
of vertices of these graphs.
Once we have results for $G^{\prime}$, then we can easily state the
same results for $G$.
From algorithm D, we will obtain the following results: 

{\bf 1.} All isolated vertices of $G$, 

{\bf 2.} $G$ is connected or not, 

{\bf 3.} Number of connected components of $G$, vertices and order of each connected component.  

Here are the steps of the algorithm.  

{\bf D1.} [ Initialize. ] Set $i\leftarrow 1, r\leftarrow 0, A^{\prime}\leftarrow A, b\leftarrow n$.
Here $r$ gives number of isolated vertices.

{\bf D2.} If $i=n+1$, go to D5.

{\bf D3.} [ Finding an isolated vertex. ] 
If 
$\sum_{j=1}^n a_{ij} = 0$  where $a_{ij}$ is $(ij)^{th} $ element of $ A^{\prime}$, 
increase $r$ by 1, interchange $R_i$ \& $R_n$, $C_i$ \& $C_n$
where $R_i$ is $i^{th}$ row and $C_i$ is $i^{th}$ column of $A^{\prime}$. 
$n\leftarrow n-1$ and go to D2. 

{\bf D4.} [ Finding all isolated vertices. ] 
If  
$\sum_{j=1}^n a_{ij} \neq 0$, increase $i$ by 1 and go to D2. 

{\bf D5.} If $r=b$ then $l\leftarrow 0$ and algorithm terminates. 

{\bf D6.} [ Initialize. ] Set $i\leftarrow 2, p\leftarrow 2,j \leftarrow p, l\leftarrow 1,
y\leftarrow 0, l_0\leftarrow 0$.
Here $l+r$ will give number of connected components, $l_y - l_{y-1}$ will give order
of each connected component, $p$ represents the
vertex under consideration.  Let $Q$ is an empty set. 

{\bf D7.} If $i=n-1$, increase $y$ by 1, $l_y\leftarrow n$  
and terminates algorithm. 

{\bf D8.} [ Finding first non-zero entry of each row. ] For each row $j$, $j$ from $p$ to $n$,
we will calculate first non-zero entry for each row $j$, i.e.
$ k_j \leftarrow $ max $ \{ i \mid \sum_{x=1}^{i-1}  a_{jx} = 0 \} $  

{\bf D9.} [ Finding minimum of all $k_j$'s. ] Now set $Q$ will have
all $k_j$'s which are minimum of all $k_i$'s i.e. 
$ Q \leftarrow $ min $ \{ k_x \mid p \leq x \leq n \} $ 

{\bf D10.} [ Finding minimum of all $j$'s for which $k_j=Q$. ] i.e.
$s \leftarrow $ min $ \{ j \mid k_j = Q \} $ 

{\bf D11.} [ Interchanging order of labell of vertices. ]
If $(s\neq i)$, we will interchange $R_s$ and $R_i$, $C_s$ and $C_i$. 

{\bf D12.} [ Checking $G$ is connected or not. ] 
If condition $K_i$ holds i.e.  
$ \sum_{s=i+1}^n a_{is} = 0 $ and $ \forall \ t>i,\   \sum_{s=1}^i a_{ts} = 0  $  
then $G^{\prime}$ is disconnected, where $G^{\prime}$ is the graph obtained from $A^{\prime}$. Now $G^{\prime}$ and $G$ are same graphs, with only difference of labelling of vertices.
It means that if $G^{\prime}$ is disconnected, then $G$ is disconnected.
 
If $K_i$ holds, increase $l, i, y$ and $p$ by 1 and $l_y\leftarrow i$. 
Here we increase $p$ by 1 because if $K_i$ holds, $v_{y-1}$ to $v_i$ will give
a connected component and next connected component start from $v_{i+1}$ but $K_{i+1}$ can not holds,
so we increase $p$ to skip $v_{i+1}$. 

{\bf D13.} [ Finding all possible $i$'s for which $K_i$ holds. ]
Increase $p$ and $i$ by 1 and return to D7. 

\begin{theorem} 
From algorithm D 

{\bf 1.} If $r>0$ or/and $l>1$(i.e. there exist some $i$ in $A^{\prime}$ 
for which condition $K_i$ holds) then $G$ is disconnected

{\bf 2.}If 
$r=0$ and $l=1$(there does not exist some $i$ in $A^{\prime}$ 
for which condition $K_i$ holds) then $G$ is connected

{\bf 3.} $l+r$ gives the number of connected components of $G^{\prime}$
i.e. of $G$
 
{\bf 4.} $l_i - l_{i-1}$ gives the order of
each connected component, and vertices from $v_{l_{i-1}}$ to  $v_{l_{i}}$ are the vertices 
of the connected components of $G^{\prime}$ i.e. of $G$ where 
$1\leq i \leq y$ and  $v_i$ is the $i^{th}$ vertex of $G^{\prime}$.
\end{theorem}

\begin{proof} First we will show that if $r>0$, then $G$ is disconnected.

If $r>0$ then for any $i$ in $A$, we have
$\sum_{j=1}^n a_{ij} = 0$. It means that
$v_i $ is not connected to any one of the vertices in $G$ 
i.e. there does not exist any path between $v_i$ and other vertices of $G$. Hence $G$ is disconnected.

Since the Theorem has three parts, we will prove each
part one by one. First we will prove first and second part i.e. 
\\ \\
1. If $r>0$ or/and $l>1$ then $G$ is disconnected else if 
$r=0$ and $l=1$ then $G$ is connected

We will prove it by mathematical induction on $n$.
Consider that $G$ is disconnected and $n=2$.
So only possible $G$ in this case is
  
.     \ \ \ \ \ \     .
 
For this $G$, we have
$
A^{\prime} =\begin{bmatrix} 0 & 0\\0 & 0 \end{bmatrix}
$.
In $A^{\prime}$, for $i=1$ and $i=2$ we have
$\sum_{j=1}^n a_{ij} = 0$ 
 
$\implies r>0 \implies$ result is true for $n=2$. 

Let the result be true for $n=p$. 
We have to show that the result will be true for $n=p+1$.
Now since $G_p$ is disconnected, where $G_p$ is a 
simple graph when $n=p$,therefore the following cases are possible:

1. $r=0$ and $l>1$

2. $r=p$ and $l=0$

3. $r=k, 0<k<p$ and $l>0$
\\ \\
Consider
$r=0$ and $l>1$.
In this case, first we obtain $A_p^{\prime}$ from $A_p$ and 
$G_p^{\prime}$ from $A_p^{\prime}$,
where $A_p$ is the adjacency matrix of graph $G_p$ and $A_p^{\prime} $ is obtained from $A_p$ by applying algorithm D.
When we add a vertex to $G_p^{\prime}$, 
we will obtain a graph say $G_{p+1}$. From
$G_{p+1}$, we can obtain $A_{p+1}$ and then by applying algorithm D to 
$A_{p+1}$, we will get $A_{p+1}^{\prime} $ and from $A_{p+1}^{\prime} $
we will get $G_{p+1}^{\prime}$.
We can obtain $G_{p+1}^{\prime} $ from $G_p^{\prime} $ in following two ways:

{\it 1. By adding a vertex and no edge i.e. a vertex of degree zero}

{\it 2. By adding a vertex and $k>0$ edges  i.e. a vertex of degree $k$} 
\\ \\
{\it 1. By adding a vertex and no edge} 

Let  $v_x$ be the added vertex. 
In this case for $i=x$ we have
$\sum_{j=1}^n a_{ij} = 0$ 
$\implies G $ is disconnected 
$\implies$ result is true for $n=p+1$.
\\ \\
{\it 2. By adding a vertex and $k$ edges} 

Now since result is true for n=p,\   therefore in $A_p^{\prime} $, 
there exist some $s$  for which $K_s$ holds.
To obtain $G_{p+1}$, this vertex  of degree $k$ can be added to $G_p^{\prime}$ 
in 3 different ways:

{\sl a. 
Let $v_x$ be adjacent to $k$ vertices which correspond  
to any of the $k$ rows from $1$ to $s$ in $A_p^{\prime}$ .}

{\sl b. 
Let $v_x$ be adjacent to $k$ vertices which correspond  
to any of the $k$ rows from  $s+1$ to $n$ in $A_p^{\prime}$ .}  

{\sl c. Let $v_x$ be adjacent to $k$ vertices which correspond  
to any of the $k-r$ rows from $1$ to $s$ and $r$ rows from  $s+1$ to $n$ in $A_p^{\prime}$. } 
\\ \\
{\sl a. Let $v_x$ be adjacent to $k$ vertices which correspond  
to any of the $k$ rows from $1$ to $s$ in $A_p^{\prime}$.}

When we add $v_x $ to $G_p^{\prime}$ and apply algorithm D,\   
because of $R_x,\   s$ in $A_p^{\prime} $ becomes $s+1$ 
in $A_{p+1}^{\prime} $.
Therefore we only need to show that $ K_{s+1} $ holds for $A_{p+1}^{\prime} $. 
In $A_p^{\prime} $ there exist some $s$ for which $K_s$ holds and 
since $k$ vertices are adjacent to $v_x$ in $G_{p+1}^{\prime}$  
$\implies$ each of these $k$ vertices are adjacent to any 
one of the $k$ vertices in $G_p^{\prime}$,\   numbered from $1$ to $s$ 
$\implies$ these $k$ edges are not adjacent to any one of the vertices 
in $G_p^{\prime}$ numbered from $s+1$ to $n$   
$\implies K_{s+1}$ holds for $A_{p+1}^{\prime} $ 
$\implies$ result is true for $n=p+1$.
\\ \\
{\sl b. Let $v_x$ be adjacent to $k$ vertices which correspond  
to any of the $k$ rows from  $s+1$ to $n$ in $A_p^{\prime}$ .  } 

Since $k$ edges are added to vertices corresponding to rows from $s+1$ to $n$,\   therefore 
after addition of vertex $v_x $,\   $s$ in $A_p^{\prime} $ will remain $s$ in $A_{p+1}^{\prime} $.
Therefore we only need to show that $K_s$ holds for $A_{p+1}^{\prime} $.
Now in $A_p^{\prime} $ there exist some $s$  for which $K_s$ holds.
Since $k$ vertices are adjacent to $v_x$ in $G_{p+1}^{\prime}$ 
$\implies$ these $k$ edges are adjacent to any one of $k$ vertices in $G_p^{\prime}$ numbered from $s+1$ to $n$   
$\implies$ these $k$ edges are not adjacent to any one of $k$ vertices in $G_p^{\prime}$ numbered from $1$ to $s$ 
$\implies K_s$ holds in $A_{p+1}^{\prime} $ for some $s$ 
$\implies$ result is true for $n=p+1$.
\\ \\
{\sl c. Let $v_x$ be adjacent to $k$ vertices which correspond  
to any of the $k-r$ rows from $1$ to $s$ and $r$ rows from  $s+1$ to $n$ in $A_p^{\prime}$.  } 

Since result is true for $n=p $ 
$\implies $ there exist some $s_1$ in $A_p^{\prime} $ for which $K_{s_1}$ holds.
Now after addition of $v_x$ in $G_p^{\prime}$,\   $v_x$ be adjacent to $k$ vertices which correspond  
to any of the $k-r$ rows from $1$ to $s$ and $r$ rows from  $s+1$ to $n$ 
$\implies$ for this $s_1$ either 
$ \sum_{p=s+1}^n a_{sp} \neq 0$  
or 
there exist atleast one $t>s$ for which  
$\sum_{p=1}^{s} a_{tp} \neq 0$ 
$\implies$ for this $s_1,\   K_{s_1}$ does not hold
$\implies$ there does not exist some other $s_i$ in $A_p^{\prime} $ for which $K_{s_i} $ holds.   
$\implies$ result is true  for $n=p+1$.
If there exist some other $s_i$'s in $A_p^{\prime} $ for which $K_{s_i}$ holds 
then for these $s_i$'s, case (c) will become case (a) or case (b) and 
since result is true for case (a) and case (b)  
$\implies$ result is true for case (c) also    
$\implies$ result is true  for $n=p+1$.
 \\ \\
Now  we prove the result when $r=p$ and $l=0$. 

For this case we add a vertex and $k$ edges to $G_p^{\prime}$.
If $k<p$ then $r>0$ for $G_{p+1}^{\prime}$.
If $k=p$ then $r=0$ and $l=1$ for $G_{p+1}^{\prime}$. 
Hence result is true for $n=p+1$. 
\\ \\
Now  we prove the result when when $r=k, 0<k<p$ and $l>0$. In this case,

a. if we add a vertex of degree $k_1(<k)$  to $G_p^{\prime}$,
then $r>0$ for $G_{p+1}^{\prime}$.  Hence result is true for $n=p+1$. 

b. if we add a vertex of degree $k_1(\geq k)$ to $G_p^{\prime}$,
then there are following possibilities:

{\sl 1. $l$ get increased by one and $r=0$. } 

{\sl 2. $l$ get decreased by one and $r>0$. } 

{\sl 3. $l$ remains same and $r=0$. } 

{\sl 4. $l$ remains same and $r>0$. } 

Clearly, in all the cases result is true for $n=p+1$. 
\\ \\
Till now we have proved the first half part of first part
of the Theorem. Now we will prove remaining half part. To prove 
this let us consider that $G$ is connected. 
We will prove the result by mathematical induction on $n$. Let n=2.
So possible $G$ for this case is   

1. --------------------- 2.
 
Now 
$
A=\begin{bmatrix} 0 & 1\\1 & 0 \end{bmatrix}
$.
Here there does not exist any s in $A$ for which $K_s $ holds 
$\implies $ result is true for $n=2$. 
Let the result be true for $n=p$.  
We have to show that result will be true for $n=p+1$.
We can obtain $G_{p+1}^{\prime} $ from $G_p^{\prime} $ in two ways:
 
1. By adding a vertex and no edge 

2. By adding a vertex and $k>0$ edges
 \\ \\
1. By adding a vertex and no edge

Let vertex $v_x$ be added. 
In this case for $i=x$ we have
$\sum_{j=1}^n a_{ij} = 0$ 
$\implies r>0 $ 
$\implies$ result is true for $n=p+1$.
\\ \\
2. By adding a vertex and $k$ edges

Since $G_p^{\prime} $ is connected  
$\implies $ there does not exist any s in $A_p^{\prime} $  for which $K_s$ holds 
$\implies $ for each $s$ in $A_p^{\prime} $ either  
$ \sum_{p=s+1}^n a_{sp} \neq 0$ 
or 
there exist atleast one $t>s$ for which 
$\sum_{p=1}^{s} a_{tp} \neq 0$. 
Now let vertex $v_x $ is added to $G_p^{\prime} $ to obtain $G_{p+1} $ and $v_x$ is adjacent to $k$ edges.
So after addition of vertex $v_x $ for each $s$ in $A_{p+1}^{\prime} $ other than $x$ either 
$ \sum_{p=s+1}^{n+1} a_{sp} \neq 0$  
or 
there exist atleast one $t>s$ for which  
$\sum_{p=1}^{s} a_{tp} \neq 0$.  
Now if $x<n+1$ then 
there exist atleast one $t>x$ for which  
$\sum_{p=1}^{x} a_{tp} \neq 0$  
and if $x=n+1$ then  
$ \sum_{p=x+1}^{n+1} a_{sp} \neq 0$  
So result is true for n=p+1. 
Here we have proved second half of first part of the Theorem.
Now will prove second part of the Theorem i.e. $l+r$ gives number of connected components 
\\ \\
We know that if there exist any $s$ for which $K_s$ holds
then $G$ is disconnected. In algorithm D, $l$ calculates the 
number of $s$ for which $K_s$ holds i.e. 
$l = 1+ $ number of $s$ for which $K_s $ holds.
Also $r$ calculates the number of isolated vertices 
$\implies l + r $ gives the number of connected components. 
\\ \\ 
Now will prove last part of the Theorem i.e. $l_i - l_{i-1}$, where $1\leq i \leq y$ gives the order of
each connected component and vertices from $v_{l_{i-1}}$ to  $v_{l_{i}}$, where 
$1\leq i \leq y$, gives the vertices of each connected component of $G^{\prime}$. 
Since we know that, when there exist some $s$ for which $K_s$ holds,
then $G$ is disconnected and we increase $y$ by one at this point.
So from $l_{y-1}$ to $l_y$, there exist a connected component and vertices between
$l_{y-1}$ to $l_y$ gives the vertices of corresponding connected component.
Hence in this way we will obtain order of all connected components and their
corresponding vertices.
\end{proof}

\subsection{Example}
Let us consider the graph shown in Figure 1A with 5 vertices. For Figure 1A, the adjacency matrix is given as
$$
A=\begin{bmatrix}
0 & 1 & 0 & 1 & 0 \\
1 & 0 & 0 & 1 & 0 \\ 
0 & 0 & 0 & 0 & 1 \\
1 & 1 & 0 & 0 & 0 \\
0 & 0 & 1 & 0 & 0 
\end{bmatrix}
$$

\begin{figure}
	\centering
	\includegraphics[width=0.65\textwidth]{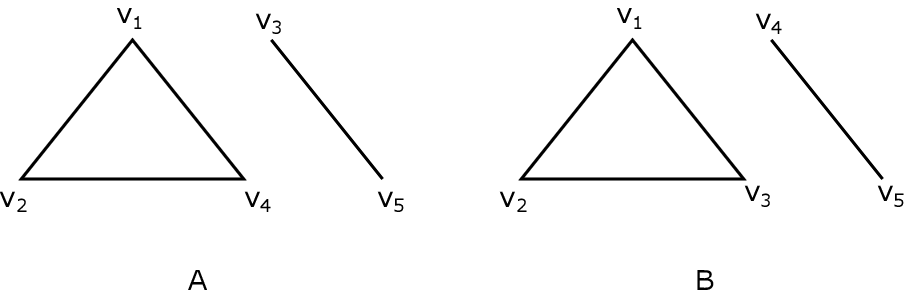}
	\caption{}
\end{figure}

To get the desired results, we follow the steps of algorithm D. 

Initially $n=5$ and $A=A^{\prime}$. 
Set $i\leftarrow 1, r\leftarrow 0$.

Now in $A^{\prime}$, there does not exist any $i$ for which  
$\sum_{j=1}^n a_{ij} = 0$  $\implies r=0$.
 
Set $i\leftarrow 2, p\leftarrow 2,j \leftarrow p, l\leftarrow 1,
y\leftarrow 0, l_0\leftarrow 0, Q\leftarrow \emptyset$. 
Now $k_2=1,\   k_3=5,\   k_4=1,\   k_5=3$ and $Q=\{ 1 \} $. 
Here $s=2=i $ but $K_2$ does not hold.  

Again $p=3, i=3 $ and $k_3=5,\   k_4=1,\   k_5=3 
\implies Q=\{ 1 \} \implies s=4$. 
Now $i=3$ i.e. $s \neq i$ 
$\implies$ Interchange $R_3 $ and $ R_4$,\ $ C_3$ and $ C_4$. We have
$$
A^{\prime}=\begin{bmatrix}
0 & 1 & 1 & 0 & 0 \\
1 & 0 & 1 & 0 & 0 \\ 
1 & 1 & 0 & 0 & 0 \\
0 & 0 & 0 & 0 & 1 \\
0 & 0 & 0 & 1 & 0 
\end{bmatrix}
$$

Now $K_3$ holds $\implies l++, i++, y++, p++$ and $l_y\leftarrow i$.
Again $k_4=5,\   k_5=3 \implies Q=\{ 3 \} \implies s=5$. 
Now $r=4$ and $s \neq r$ 
$\implies$ Interchange $R_5 $ and $R_4 $,\   $C_5 $ and $ C_4$.
Here $K_4$ does not hold. 
Algorithm ends here. Now $y++$ and $l_y\leftarrow n$. 
Here $K_3$ holds implies graph is disconnected and $l=2$ gives the number
of connected components. First connected component consists of vertices $v_1, v_2, v_3$ of $G^{\prime}$ and second
connected component consists of vertices $v_4, v_5$ of $G^{\prime}$. 
$G^{\prime}$ is shown in Figure 1B.

\end{document}